\documentclass[journal]{IEEEtran}

\usepackage{blindtext}
\usepackage{relsize}
\usepackage{amsmath}
\usepackage{color,soul}

\usepackage{amssymb}
\usepackage{amsthm}
\usepackage{amsfonts}
\usepackage[space,noadjust]{cite}
\usepackage{xspace}
\usepackage{times} 

\usepackage{algorithm} 
\usepackage[noend]{algpseudocode}
\newtheorem*{theorem}{Theorem}
\newtheorem*{proposition}{Proposition}

\newcommand{\res}[1]{\textit{Resource#1}\xspace}

\newcommand{\algoAvgCost}{ AVERAGE-COST-ALLOCATION\xspace}   
\newcommand{\algoRandInit}{ RAND-INIT-ALLOCATION\xspace}   
\newcommand{\allocate}{ ALLOCATE\xspace}
\newcommand{\nonSplitCost}{ NON-SPLIT-COST\xspace }
\newcommand{\splitCost}{ SPLIT-COST\xspace }
\newcommand{\CAL}[1]{\mathcal{#1}\xspace}

\usepackage{cite}

\usepackage{url}
\usepackage{textcomp}

\usepackage{textcomp}
\usepackage{bbm}

\usepackage{lipsum}
\usepackage{csquotes}

\newlength\mylen
\usepackage{graphicx}
\usepackage[caption=false,font=footnotesize]{subfig}

\usepackage{authblk}

\makeatletter
\newcommand{\removelatexerror}{\let\@latex@error\@gobble}
\makeatother

\begin{document}
\title{Allocation of Heterogeneous Resources of an IoT Device to Flexible Services}

\author{\IEEEauthorblockN{Vangelis Angelakis, Ioannis Avgouleas, Nikolaos Pappas, Emma Fitzgerald, and Di Yuan}
		
	\thanks{V. Angelakis, I. Avgouleas, N. Pappas, D. Yuan are with the Department of Science	and Technology, Link{\"o}ping University, SE-60174 Norrk{\"o}ping, Sweden.
	(emails: \{vangelis.angelakis, ioannis.avgouleas, nikolaos.pappas,  di.yuan\}@liu.se). D. Yuan is also visiting professor at the Institute for Systems Research, University of Maryland, College Park, MD 20742, USA.} 
	\thanks{ E. Fitzgerald is with the Department of Electrical and Information Technology, Lund University, SE-221 00 Lund, Sweden. 
	email: emma.fitzgerald@eit.lth.se}
	\thanks{Early results of this work have been presented in the IEEE INFOCOM 2015 Student Workshop \cite{SIA_INFOCOM}.}
	\thanks{This work was supported in part by the Excellence Center at Link{\"o}ping-Lund in Information Technology. Furthermore, the research leading to these results has also received funding from the European Union's Seventh Framework Programme (FP7/2007-2013) under grant agreements n° 324515 (MESH-WISE), 612316 (SOrBet), and 609094 (RERUM).}
}

\maketitle
\IEEEpeerreviewmaketitle
  
\begin{abstract}
Internet of Things (IoT) devices can be equipped with multiple heterogeneous network interfaces.
An overwhelmingly large amount of services may demand some or all of these interfaces' available resources. 
Herein, we present a precise mathematical formulation of assigning services to interfaces with heterogeneous resources in one or more rounds. 
For reasonable instance sizes, the presented formulation produces optimal solutions for this computationally hard problem. We prove the NP-Completeness of the problem and develop two algorithms to approximate the optimal solution for big instance sizes.
The first algorithm allocates the most demanding service requirements first,  considering the average cost of interfaces resources. The second one calculates the demanding resource shares and allocates the most demanding of them first by choosing randomly among equally demanding shares.
Finally, we provide simulation results giving insight into services splitting over different interfaces for both cases.
\end{abstract}
\begin{IEEEkeywords}
	Resource Management, Optimization, Internet of Things, Network interfaces, Scheduling Algorithms, Integer Linear programming.
\end{IEEEkeywords}

\section{Introduction}

\IEEEPARstart{O}{} {ver} the last few years we have witnessed the technological revolution represented by the Internet of Things (IoT) \cite{Atzori_2010}. A massive number of devices with different capabilities such as sensors, actuators, smart objects, and servers can interconnect and give rise to the development of compelling services and applications. Each IoT device can be perceived as an edge-node of a cyber-physical ecosystem with the ability to dynamically cooperate and make its resources available in order to reach a complex goal i.e., the execution of one or more tasks assigned to the network \cite{Bassi2013ett}.

Although available resources (exchangeable energy, processing power, storage capabilities etc.) are often limited, IoT devices may be called on to provide a large variety of services. It is evident that an efficient allocation of these IoT resources would improve the performance of this network. Optimal resource allocation for IoT is not trivial considering its distributed and heterogeneous nature. 

In this paper we assume that an IoT device consists of multiple network interfaces of heterogeneous technologies. We consider that each of them has a set of non-interchangeable resources which are in demand by a given set of services. Considering that the services are flexible in that they can be split between more than one interface, we model the assignment of them to the interfaces (see \figurename{\ref{fig:IoT_example}}). We call this problem \textit{Service-to-Interface Assignment (SIA)} and characterize its complexity. We also provide fast algorithms which we compare with the optimal solution.

\begin{figure}[!t]
	\vspace{-25pt}
	\begin{center}
	\includegraphics[width=1.1\columnwidth]{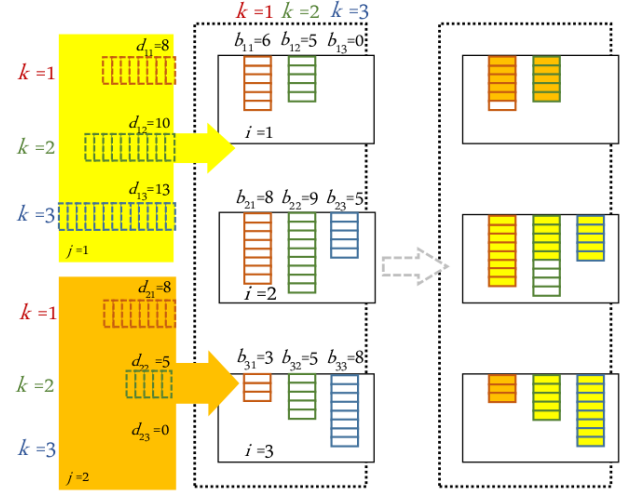}
	\end{center}
	\vspace{-20pt}
	\caption{An instance of an allocation on an IoT device with three interfaces ($i={1,2,3}$) offering three different resources ($k={1,2,3},$ visualized with red, green, and blue in the center of the figure). \textit{Left:} Before the allocation, two services demand resources; the first (yellow), and the second (orange) services demand $d_{1} = (8, 10, 13)$ and $d_{2} = (8, 5, 0) $ of these resources respectively. \textit{Right:} After the allocation, two service splits have happened: the demands of the first (yellow) service are served by the second and the third interface, while the second (orange) service's resources are split into two interfaces; the first and the third one. Note the leftover capacities of the allocation: one unit of the first (red) resource of the first interface and four units of the second (green) resource of the second interface. }
	\label{fig:IoT_example}
\end{figure}

\subsection{Related Work}
Heterogeneity of networking systems with multiple and different interfaces has been extensively studied. 
Existing applications in 4G/Wi-Fi heterogeneous wireless networks, related to this work, are presented in \cite{Cesana_VTC_2015}. Therein, the authors present methods of offloading tasks such as object recognition in a series of images, features matching, and computing descriptors to accelerate the required tasks and avoiding traffic overload. Similarly, the authors of \cite{Capone_INFOCOM_2014} propose a framework to disseminate content such as social network updates, notifications, and news feeds.
A framework for reliable networking, considering also Free Space Optic (FSO) connections and Optical Fibers, is presented in \cite{Pappas_JSAC}.

Resource Allocation (RA) has been extensively studied in wireless networks. In \cite{Tassiulas_NOW}, models that capture the cross-layer interaction from the physical to transport layers in wireless network architectures such as cellular, ad-hoc and sensor networks as well as hybrid topologies are presented. These models consider the time varying nature of arbitrary networks with different traffic forwarding modes, including datagrams and virtual circuits. Information and Communications Technology (ICT) systems may allocate resources according to some performance metrics such as efficiency and fairness \cite{Jain_idx_1998, Chiang_Multi-Rerource_Allocation_Tradeoffs}, or based on traffic type \cite{Radio_Res_Alloc_JSAC_2012}. Such objectives can often be conflicting and it may be very hard to simultaneously achieve them to a satisfactory degree \cite{Chiang_Multi-Rerource_Allocation_Tradeoffs}.
Ismail \textit{et al.} have also published a work where constant and variable bit rate services with priorities are considered \cite{Res_Alloc_Ismail_JSAC_2012}, such that each network gives a higher priority on its resources to its own subscribers rather than other users.

An example IoT application with strict demand constraints over resources to serve multiple users is presented in \cite{Res_Alloc_Delay_Constr_JSAC_2014}. The proposed framework is also appropriate to many RA settings because it provides ways to coordinate users with delay preferences using only limited information.
A distributed protocol based on consensus is proposed in [17] to solve the RA problem and management in IoT heterogeneous networks. This protocol is robust to link or node failures in dynamic scenarios.
In that paper, an IoT scenario is presented with nodes pertaining to a given IoT task by adjusting their task frequency and buffer occupancy.

Variable channel conditions and users' demands tie RA in with Quality of Service (QoS) requirements and guarantees. The research has offered different angles on tackling RA. For instance, Tan \textit{et al.} \cite{Tan_Zhu_Fei_Ge_Xiong_Utility_Max} present methods and algorithms to maximize network utility for three different QoS requirement scenarios. The traditional QoS attributes such as throughput, delay and jitter are not necessarily suited to IoT, but can still be relevant depending on the application \cite{Fagih_2013,He_2014,Zhao_2011}. Thus, the QoS in IoT is still not well-defined, mainly because an IoT service cannot be defined as the simple acquisition and processing of information and the decision making process in identification and communication \cite{Guo_2012}. In IoT, more QoS attributes such as \textit{information accuracy, privacy}, and \textit{timeliness} which rely fundamentally on the \textit{reliability of the network connectivity, availability of device energy}, and \textit{overall resources} may considerably matter \cite{Age_of_Info_ICC_2015, Age_of_Info_ISIT_2013}.

Additionally, some IoT services are required to be reconfigurable and composable. Li \textit{et al.} \cite{QoS_Scheduling_II_2014} propose a three-layer QoS scheduling model for QoS-aware IoT service. At the application layer, the QoS schedule scheme explores optimal service composition by using the knowledge provided by each service. Contemporary applications of RA can be found in \cite{CAO_2_SENSORS_JOURNAL}, where the impact of inter-user interference in Wireless Body Sensor Networks (WBSNs) is studied. 

The heterogeneity of IoT devices and the resources they provide to developing IoT applications are at the core of our work. Recent architectural frameworks (see e.g., \cite{RERUM01,RERUM02} and the references therein) call for de-verticalization of solutions, with applications being developed, independently of the end devices, which may be anything from a sensor to the latest smartphone.

We focus on IoT networking devices having multiple, different interfaces, each of which has access to a collection of finite heterogeneous resources such as downlink data rate, buffer space, CPU interrupts, and so forth. We also consider that each service is characterized by a set of demands that can be served by the resources available on one device's interfaces. 
Assuming a middleware has already assigned a service onto a given device, in this work we address the problem of flexibly mapping the service resource demands onto the interfaces of that device. 
The novel notion of flexibility of the services lies in the assumption that a demand may be served by more than one of the available interfaces, in case the available resource does not suffice, or the cost of utilizing resources over different physical interfaces proves beneficial.
From the point of view of a service, the mapping of resources from the device's interfaces to its demands could be viewed as a virtual serving interface.

\subsection{Contribution and Paper Layout}
In this work, we present a mixed-integer linear programming (MILP) formulation of the problem of assigning services to interfaces with heterogeneous resources. 
The goal is to minimize the total cost of utilizing the interfaces' resources, while satisfying the services' requirements. We consider the total cost as the sum of the cost of utilizing each resource unit along with the activation of each interface being engaged to serve a service.
An example of an instance of this assignment problem can be seen in \figurename~\ref{fig:IoT_example}. 
We also prove the NP-Completeness of the problem. For reasonable instance sizes, the presented formulation produces optimal solutions. 
We present two cases; (i) when the interfaces have enough available resources to serve the demands in one round, and (ii) when the interfaces need to serve the demands in multiple rounds i.e., serving a partial amount of the demands in each round. 
We develop two algorithms to approximate the optimal solution for large instance sizes. The first algorithm allocates the most demanding service requirements first, taking into consideration the average cost of utilizing interface resources, while the second one calculates the demanding resource shares and allocates the most demanding of them first by choosing randomly among equally demanding shares.
Finally, we provide simulation results giving insight into service splitting over different interfaces for both cases.

The rest of the paper is organized as follows. In Section \ref{sec:SystemModel}, we describe the MILP formulation of assigning services to interfaces with heterogeneous resources, where the allocation takes place in one round i.e., the interfaces capacities can accommodate the whole resource demands and the problem is feasible. Herein lies the NP-Completeness proof of the problem. 
In Section \ref{sec:Algo_Solution}, we analyze the algorithms we derived to approximate the optimal solution to the problem. 
Section \ref{sec:SchedulingSection} provides the extension of the problem to more than one round. Additionally, we prove a proposition for the number of rounds required to ensure feasibility. In Section \ref{sec:Results}, we present the results of our experiments for each of the aforementioned cases and for different configurations of services and interfaces' parameters. Finally, Section \ref{sec:Conclusion} concludes the paper.

\section{System Model}\label{sec:SystemModel}

\begin{table}[!t]
	\renewcommand{\arraystretch}{1.2}
	\caption{Notation}
	\centering
	\begin{tabular}{ | c | l | }
		\hline
		\bfseries Symbol & \bfseries Description\\
		\hline
		$\mathcal{K}$ & the set of $K$ resources\\
		$\mathcal{I}$ & the set of $I$ interfaces\\
		$\mathcal{J}$ & the set of $J$ services\\
		$x_{ijk}$     & amount of resource $k$ on interface $i$ used by service $j$\\
		
		$c_{ik}$ & the unit utilization cost of resource $k$ on interface $i$\\ 
		
		$F_{i}$ & the activation cost of interface $i$\\
		
		$A_{ij}$ & binary indicator of $i$-th interface's activation for service $j$\\
		
		$d_{jk}$ & $j$-th service's demand for resource $k$\\
		
		$b_{ik}$ & $k$-th resource capacity of interface $i$ \\		
		
		$a_{ijk}$ & the overhead for utilizing resource $k$ on interface $i$ \\  & for service $j$\\
		\hline
	\end{tabular} \label{table:notation}	
\end{table}

We consider a set of interfaces $\mathcal{I} = \{1,\dots,I\}$. The interfaces are characterized by a set $\mathcal{K} = \{1,\dots,K\}$ of resources associated with them (for example {CPU cycles, downlink data rate, buffer size}). We assume that each service $j \in \mathcal{J} = \{1,\dots,J\}$ is associated with a $K$-dimensioned demand integer vector $\mathbf{d_j}$. Likewise, each interface has a $K$-dimensioned capacity (resource availability) integer vector $\mathbf{b_i}$. 

We consider the case that services are flexible in the sense that they can be realized by splitting their demands on multiple interfaces. To model the burden that is imposed on the operating system of the device to handle such splits, we introduce a fixed-cost factor: the activation cost per interface. This is employed as a parameter to gauge the number of splits. Aside from this fixed cost, we also consider a utilization cost per unit of resource used on an interface.

Finally, to state the problem, we make the assumption that the given assignment is feasible i.e., the interface capacities are enough to serve the requested demands, which can be expressed as \[ \sum\limits_{\mathnormal{j} \in \mathcal{J}} \sum\limits_{\mathnormal{k} \in \mathcal{K}} d_{jk} \leq \sum\limits_{\mathnormal{i} \in \mathcal{I}} \sum\limits_{\mathnormal{k} \in \mathcal{K}} b_{ik}, ~\forall \mathnormal{k} \in \mathcal{K}. \]. In Section \ref{sec:SchedulingSection}, this assumption is removed for problem extension.

Our goal is to serve all demands by assigning them to the physical interfaces, minimizing the total cost of using them, namely the total resource utilization and activation cost. We call this the \textit{Service-to-Interface Assignment (SIA)} problem. 

In the model that follows, we use the variable $x_{ijk}$  for the amount of the $k$-th resource of the $i$-th interface utilized by service $j$. 
We consider these values to be integer such as the ones in the demands vectors. We denote by $c_{ik}$ the per-unit cost to utilize resource $k$ on interface $i$. 

The activation cost of interface $i$ is $F_{i}$ and the auxiliary variable $A_{ij}$ becomes one if and only if there is at least one resource utilizing interface $i$ for service $j$.

We also assume that each service $j$ incurs an overhead on the resource it utilizes, which may vary by interface in order to capture MAC and PHY layer practical considerations; this is denoted $a_{ijk}$. 
Thus, our model amounts to:

\begin{equation}
\begin{aligned}
& \underset{}{\text{min.}}
& & \sum\limits_{\mathnormal{k} \in \mathcal{K}} \sum\limits_{\mathnormal{i} \in \mathcal{I}} \mathnormal{c}_{ik} \sum\limits_{\mathnormal{j} \in \mathcal{J}} \mathnormal{x}_{ijk} + \sum\limits_{\mathnormal{i} \in \mathcal{I}} \sum\limits_{\mathnormal{j} \in \mathcal{J}} \mathnormal{F}_{i}\mathnormal{A}_{ij},
\end{aligned}\label{eq:minOptFormulation} 
\end{equation}\vspace{-3mm}
\begin{equation}
\begin{aligned}
& \text{s.t.}
& & \sum\limits_{\mathnormal{i} \in \mathcal{I}} x_{ijk} = d_{jk}, ~\forall \mathnormal{j} \in \mathcal{J}, ~\forall \mathnormal{k} \in \mathcal{K},
\end{aligned}\label{eq:demandsConstraints} 
\end{equation}\vspace{-3mm}
\begin{equation}
\begin{aligned}
& & & \sum\limits_{\mathnormal{j} \in \mathcal{J}} (1+a_{ijk})x_{ijk} \leq b_{ik}, ~\forall \mathnormal{i} \in \mathcal{I}, ~\forall \mathnormal{k} \in \mathcal{K},
\end{aligned}\label{eq:capacityConstraints}
\end{equation}\vspace{-3mm}
\begin{equation}
\begin{aligned}
& & & x_{ijk} \geq 0, ~\forall \mathnormal{i} \in \mathcal{I}, ~\forall \mathnormal{j} \in \mathcal{J}, ~\forall \mathnormal{k} \in \mathcal{K},
\end{aligned}\label{eq:nonNegativityConstraints}
\end{equation}\vspace{-3mm}
\begin{equation}
\begin{aligned}
& & & A_{ij} = \mathbbm{1} \Big(  \sum\limits_{\mathnormal{k} \in \mathcal{K}}  x_{ijk} \Big) , ~\forall \mathnormal{i} \in \mathcal{I}, ~\forall \mathnormal{j} \in \mathcal{J},\label{eq:ACT} 
\end{aligned}
\end{equation}

where the objective of \eqref{eq:minOptFormulation} is to minimize the total cost of two terms: the first aims to capture the total cost incurred by the utilization of the resources over heterogeneous interfaces, while the second term captures the cost introduced by splitting the service over multiple interfaces, since with each additional interface utilized the overall cost is encumbered by another activation cost $F$-term. The set of constraints in \eqref{eq:demandsConstraints} ensures that all services demands are met, while the constraints of \eqref{eq:capacityConstraints} ensure that the service allocation will be  performed on interfaces with available resources. 
In \eqref{eq:ACT} the $\mathbbm{1}(.)$ symbol denotes the characteristic function becoming one if the argument is positive, zero otherwise, thus, $A_{ij}$ becomes one if one resource of service $j$ is served by interface $i$. 
A summary of the notation we use can be found in \tablename{~\ref{table:notation}}.

\begin{theorem} The SIA is NP-Complete. \end{theorem}\label{thm:SIA}

\begin{proof}
	The \textit{Partition Problem (PP)} amounts to determining if a set of $\mathnormal{N}$ integers, of sum $\mathnormal{S}$ can be partitioned into two subsets, each having a sum of  $\mathnormal{S}/2$. The PP is a well-known NP-Complete problem \cite{Karp}. We base the proof on the construction of an instance of the problem from any instance of the Partition Problem, as follows. 
	
	Assume that we have only one resource type on the interfaces available $(K=1)$. Let each element in the set of the PP be a service of the SIA problem instance and the value of each element be the resource demand $d_j$ of the corresponding $j$-th service. Additionally, let there be just two interfaces $ (I=2) $, each with resource availability $ b_i=\mathnormal{S}/2, \forall i\in \{1,2\} $. 
	We set  the overhead coefficients to zero $\alpha_{ij}=0, \forall i  \in \{1, 2\}, \forall j \in \{ 1,\dotsc,J \} $ and the utilization cost to zero likewise $ c_{ij} = 0 $, while we fix the interface activation cost to one: $ F_{ij} = 1 $.
	
	The constructed SIA instance is feasible, because (i) by construction the total resource availability on the two interfaces suffices to serve the total demand and (ii) splitting service demand on more than one interfaces is allowed.
	
	Consider a solution of the constructed SIA instance where no splits occur. If such a solution exists, then each service is assigned to one interface and by \eqref{eq:minOptFormulation} the cost will be equal to $J$. Furthermore, since any service split in two interfaces gives a cost of two, if splits exist in a solution, the cost will be at least $J+1$. Hence, by the construction of the instance, it becomes obvious that no value lower  than $J$ can be achieved. The recognition version of the SIA instance is to answer whether or not there is a solution for which the cost is at most some value, in our case $J$.
	
	In any solution of the SIA the service demand assigned on each interface will be $\mathnormal{S}/2$.  If there is no split of services in a solution, then their assignment to the two interfaces is a partition of the integers in the PP with equal sum. 
	Hence, if the answer to the original PP instance is positive and we map the services to the elements of the solution subsets, then by assigning the services to the two interfaces, no splits will exist and the cost will be $J$. 
	mapping to the elements of the solution subsets
	Thus, the answer to the recognition version of the SIA instance is yes. Conversely, if the answer to the SIA is yes, then there cannot be a split service so the assignment of services to the two interfaces is a valid PP solution. Therefore, solving the constructed SIA instance is equivalent to solving an arbitrary PP instance.
\end{proof}

An example allocation has been depicted in \figurename~\ref{fig:IoT_example}. Two services demand three different resources of an IoT device's interfaces. The network interfaces offer three different resources with enough capacities in total to serve the resource demands of both services in one round. 
Note that it is possible that each resource type is not available in each interface and that a service may not demand all resource types.

Although the search space for the \textit{SIA} problem in practical implementations will be small, meaning an MILP solver would be able to provide the optimal answer in milliseconds, the assumption of the functionality of such a solver on a constrained IoT device points to the need for a fast sub-optimal algorithm. To this end, we have devised a solution with two variants that are described in Section \ref{sec:Algo_Solution}.

\section{Algorithmic Solution}\label{sec:Algo_Solution}
We have devised two algorithms which approximate the optimal solution of \textit{SIA}.
Both algorithms assign service demands to interfaces' capacities using interfaces' utilization and activation costs $c$ and $F$ respectively. They differ in the way they choose which resource demand to serve first.

The main idea of the first algorithm, which we call \algoRandInit, is to first serve the services that demand the highest resource shares. Serving the highest demands first will employ the largest amount of the least expensive interface capacities. As a result, it is highly probable that the total cost will be minimized. The algorithm then proceeds serving the demands with lower resource shares and so on. The serving order among equal resource shares is chosen randomly.  

The first four lines of the algorithm (see Algorithm \ref{alg:randInit}) calculate the resource shares of every demand. They do so by dividing each demand by the maximum resource demand of its type. The result is a normalized demands array {$\bf{d^{'}}$}. Hence, $0 \leq d^{'}_{jk} \leq 1, (~\forall \mathnormal{j} \in \mathcal{J}), (~\forall \mathnormal{k} \in \mathcal{K}) $.

Procedure RANDOM-INIT-EQUAL-SHARES in Line 5 of the algorithm takes $\bf{d^{'}}$ as input and finds which services demand equal resource shares. For these services it randomly chooses the order by which they will be served. The result is a vector ($\bf{d_s}$) with the order by which the demands will be served. 

Since there is now a data structure ($\bf{d_{s}}$) with the ordering of the demands,\algoRandInit can proceed with each demand of $\bf{d_{s}}$ beginning from the first one, and allocate it to the available interface resources. 
The initialization of two auxiliary variables follows in Lines $6-7$. The total cost is saved into variable $totalCost$. Initially, the $2D$ array $A$ consists of $I*J$ zero elements. Later, element $A_{ij}$ will be set to one if service $j$ has activated interface $i$.

The block of Lines $8-23$ allocates each service demand $\bf{d_{s}}$ to the interfaces, beginning by serving the most demanding one first. 
In the beginning, the algorithm locates which service $j$ requested resource $k$ (Line $9$) mapping the vector $\bf{d_{s}}$ to the $2D$ coordinates of $d_{jk}$. Then, it searches for the two interfaces with the lowest utilization costs per unit. The corresponding indices are saved into variables $i'$ and $i''$ (Lines $10$ and $11$). 
Next, it tries to allocate the demand to the most inexpensive --- by utilization cost per unit --- interface. 
If the interface capacities are enough (Line $12$), then the allocation happens by the\allocate procedure: the capacities are decreased appropriately by the requested demand, the $totalCost$ is updated, and the binary variable $A(i',j)$ is set to one to denote that service $j$ has activated interface $i'$. 
If the lowest-cost interface ($i'$) cannot serve $\bf{d_{s}}$, then the algorithm tries to allocate it to the interface with the next lowest utilization cost per unit ($i^{''}$), following the appropriate allocation steps (Lines $14-15$).

In the event both of these attempts fail, \algoRandInit investigates the possibility of splitting the requested service demands into two or more interfaces (Lines $16-22$).
Two costs are calculated in this case. \nonSplitCost calculates the cost of allocating the demands to the interface that can serve them fully with the lowest possible cost (Line $17$). 
\splitCost calculates the cost of splitting the demands among interfaces with available resources in descending order of cost (Line $18$). 
In case that the two lowest-cost interfaces can only partially serve the demand, the third most inexpensive interface is engaged and so forth. 
Subsequently, the minimum of these two costs (Line $19$) specifies the demand allocation (Lines $19-22$).
Finally, after every demand has been served, the activation cost is added to the $totalCost$, so that this variable stores the actual total cost of the resulting allocations (Line $23$).
 
\begin{algorithm}
	\caption{ \algoRandInit } \label{alg:randInit}
	\textbf{Input:}{~Services demands \textit{d}, interfaces utilization costs. \textit{c}, activation costs \textit{F}, and interfaces capacities \textit{b}.}\\
	\textbf{Output:}{~Services allocation to minimize utilization and activation cost of interfaces, while all services demands are satisfied.}\\
	
	\begin{algorithmic}[1]
		\For {$k = 1 \dots K $}
		\State {$ m = \max\limits_{j \in \mathcal{J}} d_{jk}$;}
		\For {$j = 1 \dotsc J$}
		\State { $ d_{jk}^{'} = d_{jk} / m $;}
		\EndFor	
		\EndFor
		\State  {$ d_s = $ {RANDOM-INIT-EQUAL-SHARES}$(d')$;}
		
		\State {$ totalCost = 0 $;}
		\State {$ A = {0}_{IXJ} $;}
		
		\For {$s:=1 \dots J*K$}
		
		\State {Let $d_{jk}$ be the demand that corresponds to $d_s$.}
		\State {$ i^{'} = \min\limits_{i\in\mathcal{I}} c_{ik}$;}
		\State {$ i^{''} = \min\limits_{i\in\mathcal{I},i \neq i^{'}}c_{ik}$;}
		\If {$ d_{jk} \leq b_{(i^{'})k }$}
		\State { $[totalCost, A ] = $ALLOCATE$(d_{jk},~ b_{(i^{'})k });$ }
		\ElsIf{$ d_{jk} \leq b_{(i^{''})k}$}
		\State { $[totalCost, A ] = $ALLOCATE$(d_{jk},~ b_{(i^{''})k });$ }
		\Else
		\State {$[nonSplitCost, l] = $ }
		\Statex {\hspace*{\mylen} {NON-SPLIT-COST} $(d_{jk}, b, c, F, A);$}
		\State { $[splitCost, l^{'}] = ${SPLIT-COST}$( d_{jk}, b, c, F, A);$}
		
		\If { $ nonSplitCost \leq splitCost $}			
		\State { $[totalCost, A ] = $ALLOCATE$(d_{jk},~ b_{lk});$ }
		\Else
		\State { $[totalCost, A ] = $ALLOCATE$(d_{jk},~ b_{ (l^{'})k });$ }
		\EndIf
		\EndIf
	\EndFor	

	\State {$ totalCost = totalCost + \sum_{i \in \mathcal{I}} \sum_{j \in \mathcal{J}} F_i A_{ij}$;}

	\end{algorithmic}
\end{algorithm}

To evaluate the effect of the randomized selection of Lines $1-5$, we also introduced a sophisticated ordering taking into account the average cost of the requested demands in deciding which resource demand to serve first.
Algorithm \ref{alg:avgCost} presents this variant, which we named\algoAvgCost.

The new algorithm differs in the order it serves the requested resource demands.
After performing demand normalization as\algoRandInit (Lines $1-4$), the\algoAvgCost calculates the average cost of each resource type by producing the dot product of (column) vectors $\bf{c_k}$ and $\bf{b_k}$, which represent $k$-th resource's utilization cost (per unit) and capacities over all interfaces respectively (Lines $5-7$).
Afterwards, element $(j,k)$ of matrix $C$ will hold the average cost of resource $k$; therefore for a fixed $k$, the value of $C_{jk}$ is the same for each $j$.

Line $8$ produces the element-wise (Hadamard) product of matrices $\bf{d^{'}}$ and $C$.
The resulting elements $d_{jk}$ of matrix $\bf{d}$ will hold the average cost of the demanded resource $k$ of service $j$. As a consequence, we can use the information of matrix $\bf{d}$ to infer the most demanding resources. This is done by reshaping the latter $2D$ matrix to a vector $\bf{d_{s}}$ and sorting it by descending order (Line $9$).

From there, the remaining part of the algorithm is the same as\algoRandInit, which sequentially processes $\bf{d_{s}}$ elements and allocates them to the interfaces \vspace{-1mm}(Lines $6-23$ of Algorithm \ref{alg:randInit}).

\begin{algorithm}
	\caption{ \algoAvgCost } \label{alg:avgCost}
	\textbf{Input:}{~Services demands \textit{d}, interfaces utilization costs. \textit{c}, activation costs \textit{F}, and interfaces capacities \textit{b}.}\\
	\textbf{Output:}{~Services allocation to minimize utilization and activation cost of interfaces, while all services demands are satisfied.}\\
	
	\begin{algorithmic}[1]
		
		\For {$k = 1 \dots K $}
		\State {$ m = \max\limits_{j \in \mathcal{J}} d_{jk}$;}
		\For {$j = 1 \dotsc J$}
		\State { $ d_{jk}^{'} = d_{jk} / m $;}
		\EndFor	
		\EndFor 
		\For {$k = 1 \dots K$}
		\For {$j = 1 \dots J$}
		\State {$ C_{jk} = (c_k \cdot b_k )/100 $;}
		\EndFor	
		\EndFor	
		
		\State {$ d = d^{'} \circ C $;}
		\State {Reshape $d$ to vector $d_s$ and sort it by descending order.}
		\State {Same as Lines $6-23$ of Algorithm \ref{alg:randInit}.}
	\end{algorithmic}
\end{algorithm}

Regarding the required computational steps, the \algoRandInit algorithm needs $2|\CAL{K}||\CAL{J}|$ steps to calculate the resource shares (Lines $1-4$), where $|.|$ is the cardinality of the included set. RANDOM-INIT-EQUAL-SHARES (Line $5$) requires at most $|\CAL{K}||\CAL{J}|log(|\CAL{K}||\CAL{J}|)$. 
Line $7$ needs $|\CAL{I}||\CAL{J}|$ steps and lines $9-22$ require $2|\CAL{I}|+2+|\CAL{I}|+|\CAL{I}|+1$ steps, since lines 10 and 11 require $|\CAL{I}|$ steps each, {ALLOCATE} is constant (1 step), and both {NON-SPLIT-COST} and {SPLIT-COST} require at most $|\CAL{I}|$ steps. Finally, the summation in line 23 requires $|\CAL{I}||\CAL{J}|$ steps.  
Therefore, lines $6-23$, which are common to both algorithms, require at most $|\CAL{I}||\CAL{J}| + |\CAL{J}||\CAL{K}|(4|\CAL{I}|+3) + |\CAL{I}||\CAL{J}| $ steps.

Lines $5-7$ of the AVERAGE-COST-ALLOCATION algorithm require $|\CAL{K}||\CAL{J}||\CAL{I}|$ steps, and line $8$ of the same algorithm needs $|\CAL{J}||\CAL{K}|$ steps to calculate the element-wise product. Moreover, line $9$ requires  $|\CAL{K}||\CAL{J}|(1+log(|\CAL{K}||\CAL{J}|))$ steps to reshape and sort the given matrix.

As a result, the first algorithm requires at most $|\CAL{K}||\CAL{J}|(4|\CAL{I}| + log(|\CAL{K}||\CAL{J}|) + 5) + 2|\CAL{I}||\CAL{J}|$ steps in total, whereas the second one requires at most $|\CAL{K}||\CAL{J}|(5|\CAL{I}| + log(|\CAL{K}||\CAL{J}|) + 7) + 2|\CAL{I}||\CAL{J}|$ steps to terminate. To conclude, both algorithmns are $\mathcal{O}(|\CAL{I}||\CAL{K}||\CAL{J}|+ |\CAL{K}||\CAL{J}|log(|\CAL{K}||\CAL{J}|))$.

\section{Allocation over Multiple Rounds}\label{sec:SchedulingSection}

In this section, we lift the assumption of feasibility in a single round. If the allocation of all demands cannot take place in a single shot (or round), we consider that we will utilize the same interface capacities for more than one round to handle the remaining resource demands. Therefore, in the first round we make an incomplete allocation (that will not serve all demands) and when these demands have been served we again employ the same interfaces' resources for a new round to serve the remaining demands repeating this process for as many rounds as necessary. 

This simple solution can be implemented by introducing an integer $R>1$ that will be the number of rounds to serve all demands with the available capacities. We call this problem \textit{multi-round SIA} and the difference from \textit{SIA} is only in (\ref{eq:capacityConstraints}), which now becomes: 

\begin{equation}
\sum\limits_{\mathnormal{j} \in \mathcal{J}} (1+a_{ijk})x_{ijk} \leq Rb_{ik}, ~\forall \mathnormal{i} \in \mathcal{I}, ~\forall \mathnormal{k} \in \mathcal{K}.
\label{eq:SchedulingCapacityConstraints}
\end{equation}\vspace{-3mm}

The selection of $R$ depends on the system designer's goal. The number of rounds are affected by the following two factors: user flexibility versus cost of the solution, and user flexibility versus duration of the solution (i.e., number of rounds). Specifically, if a designer is insensitive to the number of rounds, then in each round the lowest-cost resource will be allocated. Conversely, if the completion time is of utter importance, then fewer rounds are necessary, but the resulting cost is likely to be higher.

\subsection{Upper and lower bounds for the required number of rounds}
The minimum number of rounds, $R_{min}$, to achieve (resource demands) feasibility is given by the lower bound needed to allocate the whole set of demands. If only one round is enough, the model reduces to that of Section \ref{sec:SystemModel}. 
In order to calculate the lower bound, we choose $R$ such that the minimum possible number of rounds to allocate the whole set of demands is used. 
We find how many rounds are required to fully serve each resource type's demands and choose the maximum of these rounds. 
Therefore, $R_{min} = \max\limits_{k \in \mathcal{K}} \left \lceil \frac{D_k}{B_k} \right \rceil$, where $D_k \triangleq \sum\limits_{\mathnormal{j} \in \mathcal{J}} d_{jk}$ are the total service demands for resource $k$, $B_k \triangleq \sum\limits_{\mathnormal{i} \in \mathcal{I}} b_{ik}$  are the total capacities of resource $k$, and $\lceil.\rceil$ is used to denote the ceiling function. 
Using fewer rounds than $R_{min}$ is not enough to satisfy every demand constraint and hence the problem becomes infeasible. Lower interface capacities yield more rounds to serve the requested resource demands.

If the system designer is interested in the lowest possible total allocation cost, only the lowest-cost interfaces should be used in each round. This allocation policy clearly results in a lower total cost in comparison to the aforementioned one. However, more rounds may be used. 
The necessary number of rounds for this allocation policy, $R_{max}$, is given by choosing the number of rounds such that in each round we utilize only the lowest-cost resource and wait until it again becomes available.

Thus: $R_{max} = \max\limits_{k \in \mathcal{K}} \left \lceil \frac{D_k}{  b_{(i_{k}^{'})k} } \right \rceil$, where $i_{k}^{'} \triangleq \underset{\mathnormal{i} \in \mathcal{I}}{\arg\min}(c_{ik}D_k+F_i) $.

More than $R_{max}$ number of rounds can be used, but the total allocation cost will not be further decreased, since using this policy already uses all the available lowest-cost resources in each round.
  
The previous analysis leads to the following claim.

\begin{proposition} The number of rounds $R$ for the multi-round SIA satisfies: 
\begin{equation}  R_{min} \leq R \leq R_{max} \label{in:R} \end{equation}
\end{proposition}

\begin{proof} 
	
Assume $D_k, B_k,$ and $ i^{'}_k,$ as defined previously.
If \eqref{eq:capacityConstraints} cannot be satisfied, then the \textit{SIA} is not feasible. In this case, a workaround is to allow the allocation to take place in more than one round. In each round all the resources are available for allocation.
Hence, assume there is a positive integer $R_k>1$ such that 
$ b_{ik} < R_kb_{ik}, (\forall i \in \mathcal{I}), (\forall k \in \mathcal{K})$ is true and therefore \eqref{eq:SchedulingCapacityConstraints} holds, which means that the \textit{multi-round SIA} is feasible.
Setting $a_{ijk}=0$ in \eqref{eq:SchedulingCapacityConstraints} and decomposing the $i's$ in the last inequality yields:

\begin{equation*}
\begin{aligned}
 \sum\limits_{\mathnormal{j} \in \mathcal{J}} \sum\limits_{\mathnormal{i} \in \mathcal{I}}x_{ijk} \leq R_k\sum\limits_{\mathnormal{i} \in \mathcal{I}}b_{ik}, (\forall \mathnormal{k} \in \mathcal{K}) ~ \stackrel{\eqref{eq:demandsConstraints}}{\Leftrightarrow}
\end{aligned}
\end{equation*}
\begin{equation}
\begin{aligned}
 \frac{\sum\limits_{\mathnormal{j} \in \mathcal{J}}d_{jk}}{\sum\limits_{\mathnormal{i} \in \mathcal{I}}b_{ik}} \leq R_k, (\forall \mathnormal{k} \in \mathcal{K}) \Leftrightarrow  \frac{D_k}{B_k} \leq R_k, (\forall \mathnormal{k} \in \mathcal{K}).
\end{aligned}\label{in:R_min}
\end{equation}

Letting $R_k = \lceil \frac{D_k}{B_k} \rceil $ is enough to make \eqref{in:R_min} hold. Clearly, the maximum of those $R_{k}$'s will satisfy \eqref{in:R_min} as well. Thus, taking  $R_{min} = \underset{\mathnormal{k} \in \mathcal{K}}{\max}~R_k$ will make \eqref{in:R_min} hold and \textit{multi-round SIA} become feasible.

For the right-hand side of \eqref{in:R}, we will only use the most inexpensive interface for each resource in each round. Index $i^{'}_k$ calculates which interface offers the most inexpensive cost for the total demands of resource $k$. Therefore, only $ b_{(i^{'}_k)k}$ of every resource $k$ will be used in each round. Consequently, $ \lceil \frac{D_k}{b_{(i^{'}_k)k}} \rceil $ rounds are required to serve the whole demands ($D_k$) for resource $k$, using the most inexpensive interface in each round.
From \eqref{in:R_min} and considering the fact that $ b_{(i^{'}_k)k} \leq \sum\limits_{\mathnormal{i} \in \mathcal{I}} b_{ik} = B_k, (\forall \mathnormal{k} \in \mathcal{K})$, we have:

\begin{equation}
\begin{aligned}
\frac{D_k}{B_k} \leq \frac{D_k}{b_{(i^{'}_k)k}} \leq R^{'}_k, (\forall \mathnormal{k} \in \mathcal{K}).
\end{aligned}\label{in:R_max}
\end{equation}

Setting $ R^{'}_k = \lceil \frac{D_k}{b_{(i^{'}_k)k}} \rceil $ makes \eqref{in:R_max} hold i.e., the \textit{multi-round SIA} is feasible. The maximum of $ R^{'}_k$'s, which is $R_{max} = \underset{\mathnormal{k} \in \mathcal{K}}{\max}~R^{'}_k $, makes \eqref{eq:SchedulingCapacityConstraints} true and $ R_{min} \leq R_{max} $.
\end{proof}

Next, we will give an example to clarify the previous concepts.
\subsubsection*{Example}  Consider an IoT device with two interfaces and $D_k, B_k,$ and $ i^{'}_k,$ as previously defined.
The first and the second interface offer $b_{1} = (20, 25)$ and $b_{2} = (25,30)$ units of (\res{1}, \res{2}) respectively. Hence, the total interfaces' capacities for the two resources are: $ (B_{1}, B_{2}) = (45,55)$.
Let service demands be $ (D_1, D_2) = ( \sum\limits_{\mathnormal{j} \in \mathcal{J}}  d_{j1}, \sum\limits_{\mathnormal{j} \in \mathcal{J}}  d_{j2}) = (100, 80)$ units of these resources in total. Obviously, the total service demands cannot be accommodated in one round.

Then, the minimum number of rounds to serve the requested demands are: $ R_{min} = \max ( \lceil \frac{D_1}{B_1} \rceil, \lceil \frac{D_2}{B_2} \rceil ) =  \max ( \lceil \frac{100}{45} \rceil, \lceil \frac{80}{55} \rceil ) = 3$. Note that the calculation of $ R_{min} $ does not take into consideration interfaces' costs.

Now, consider the utilization cost of the first interface for each unit of (\res{1}, \res{2}) to be $c_{1} = (35,45)$. Similarly, the utilization cost per unit of the second interface is $c_{2} = (30,50)$. Additionally, the interfaces activation costs are $F_{1} = 100$ and $F_{2 } = 210 $ for the first and the second interface respectively.

Then, $ i_{1}^{'} = \underset{\mathnormal{i} \in \mathcal{I}}{\arg\min} ~(c_{i1}D_1+F_i) = \underset{\mathnormal{i} \in \mathcal{I}}{\arg\min} ~(3500+100, 3000+210) = \underset{\mathnormal{i} \in \mathcal{I}}{\arg\min} ~(3600,3210)= 2 $ and $ i_{2}^{'} = \underset{\mathnormal{i} \in \mathcal{I}}{\arg\min} ~(c_{i2}D_2+F_i) = \underset{\mathnormal{i} \in \mathcal{I}}{\arg\min} ~( 45*80+100, 50*80+210) = \underset{\mathnormal{i} \in \mathcal{I}}{\arg\min} ~(3700,4210)= 1 $.

Therefore, $ R_{max} = \max ( \lceil \frac{D_1}{b_{(i^{'}_1)1}} \rceil, \lceil \frac{D_2}{b_{(i^{'}_2)2}} \rceil ) = \max ( \lceil \frac{D_1}{b_{21}} \rceil, \lceil \frac{D_2}{b_{12}} \rceil ) = \max ( \lceil \frac{100}{25} \rceil, \lceil \frac{80}{25} \rceil )  = 4 $.

Note that in this example and in general the resulting allocations for $ R = R_{min}$ may leave fewer leftover capacities in comparison to $ R = R_{max}$, but produce a higher total cost since many demands are forced to be served from more expensive interfaces in order to keep $R$ as low as possible.

\section{Results}\label{sec:Results}

In this section, we first present the simulation results for the system model we described in Section \ref{sec:SystemModel}. Recall that the system model in that section accounts for the \textit{SIA} in one round. Scenarios with different sets of services and activation costs were simulated to comprehend the behavior of the system under various circumstances. These scenarios act as benchmarks to evaluate the performance of the algorithms we presented in Section \ref{sec:Algo_Solution}. The corresponding results can be found in subsection \ref{subSec:Algo_Solutions}. In the next and last subsection, we present the simulation results of \textit{multi-round SIA}. Therein, we devise a new set of simulation configurations (services and activation costs) to demonstrate the effect of the number of rounds to the total cost of the problem.

\subsection{One Round Allocation}
We performed several sets of simulations to assess the total cost and the number of splits per service for configurations of three to ten services using different interface activation costs, and services' demands.
Note that in this case \textit{SIA} is always taken to be feasible i.e., the whole set of demands is satisfied in one allocation round ($\sum\limits_{\mathnormal{j} \in \mathcal{J}} \sum\limits_{\mathnormal{k} \in \mathcal{K}} d_{jk} \leq \sum\limits_{\mathnormal{i} \in \mathcal{I}} \sum\limits_{\mathnormal{k} \in \mathcal{K}} b_{ik}, ~\forall \mathnormal{k} \in \mathcal{K} $). 
When this assumption doesn't hold, the problem will entail service allocation over multiple rounds (i.e., \eqref{eq:SchedulingCapacityConstraints} with $R>1$). Such cases are considered in Section \ref{sec:SchedulingResults}. 

Interfaces' utilization costs ($c_{ik}$'s) have been set to be constant throughout the experiments and chosen such that they are not uniform among interfaces to model practical scenarios. Additionally, the activation costs of the interfaces ($F_i$'s) were tuned in order to reflect the effect they may have on forcing the services to split among several interfaces. 

Five sets of simulation setups were considered. Using a set of three different demand classes, we produced a set of services, which we call \textit{Random Services}, by choosing randomly from one of these classes with equal probability. 
We combined the \textit{Random Services} set with three different values of activation costs to mimic actual scenarios: a \textit{High (RSH)}, a \textit{Mixed (RSM)} and a \textit{Low (RSL)} activation cost. The first two values were chosen to be an order of magnitude higher per resource type than the last one. We ran the solver 1000 times for each different \textit{Random} configuration and averaged to obtain the optimal total cost and the average number of splits per service. Two more service sets consisted of \textit{High (HDL)} and \textit{Low (LDL)} demands services along with \textit{Low} activation cost. \tablename{~\ref{table:Simulation_Services}} summarizes the setups we tested.

More services produce a total cost with a wide spread around the mean, since more \textit{High (RSH)} and \textit{Low (RSL)} services lead to additional corresponding utilization and activation costs to the total sum.  As a result, the total cost may vary significantly depending on the arrival and the heterogeneity of the services.
The plots of \figurename{~\ref{fig:Optimal_Cost_Random}} presents the total cost concerning the \textit{Random Services (RSH, RSM, and RSL)} demands scenario (for \textit{High}, \textit{Mixed}, and \textit{Low} activation cost respectively). The box-plots show the linearity of the total cost in relation to the number of services and the spread between the higher and lower value of the 1000 runs. 

The maximum total cost of \textit{Random} services is higher than \textit{HDL} services due to the higher demands of the demand	 classes from which \textit{Random} services were produced. 
In \figurename{~\ref{fig:Optimal_Cost_HighAndLow}} the total cost for \textit{RSL}, \textit{HDL}, and \textit{LDL} services is provided. The linearity of the cost is evident. Moreover, comparing to the previous figure, if the activation cost is much higher than the utilization cost, then the total cost is higher for \textit{RSL} services than for \textit{HDL} services. The average total cost for \textit{RSL} services lies between the total cost of \textit{HDL} and \textit{LDL} demands services.

\begin{table}[!t]
	\renewcommand{\arraystretch}{1.1}
	\caption{Services tested in our Experiments}
	\centering
	\begin{tabular}{ | l | l | }
		\hline
		\bfseries Services Configuration  & \bfseries Activation Cost\\
		\hline
		
		Random Services, High Activation Cost (\textit{RSH})& $F_{i}= [500, 500, 500]$ \\
		Random Services, Mixed Activation Cost (\textit{RSM})& $F_{i}= [300, 100, 200]$ \\
		Random Services, Low Activation Cost (\textit{RSL})&  $F_{i} = [20, 20, 20]$ \\
		High Demands, Low Activation Cost (\textit{HDL})& $F_{i} = [20, 20, 20]$ \\
		Low Demands, Low Activation Cost (\textit{LDL})& $F_{i} = [20, 20, 20]$ \\
		
		\hline
	\end{tabular} \label{table:Simulation_Services}
	
\end{table}

\begin{figure}
	\centering
	\includegraphics[width=1\columnwidth]{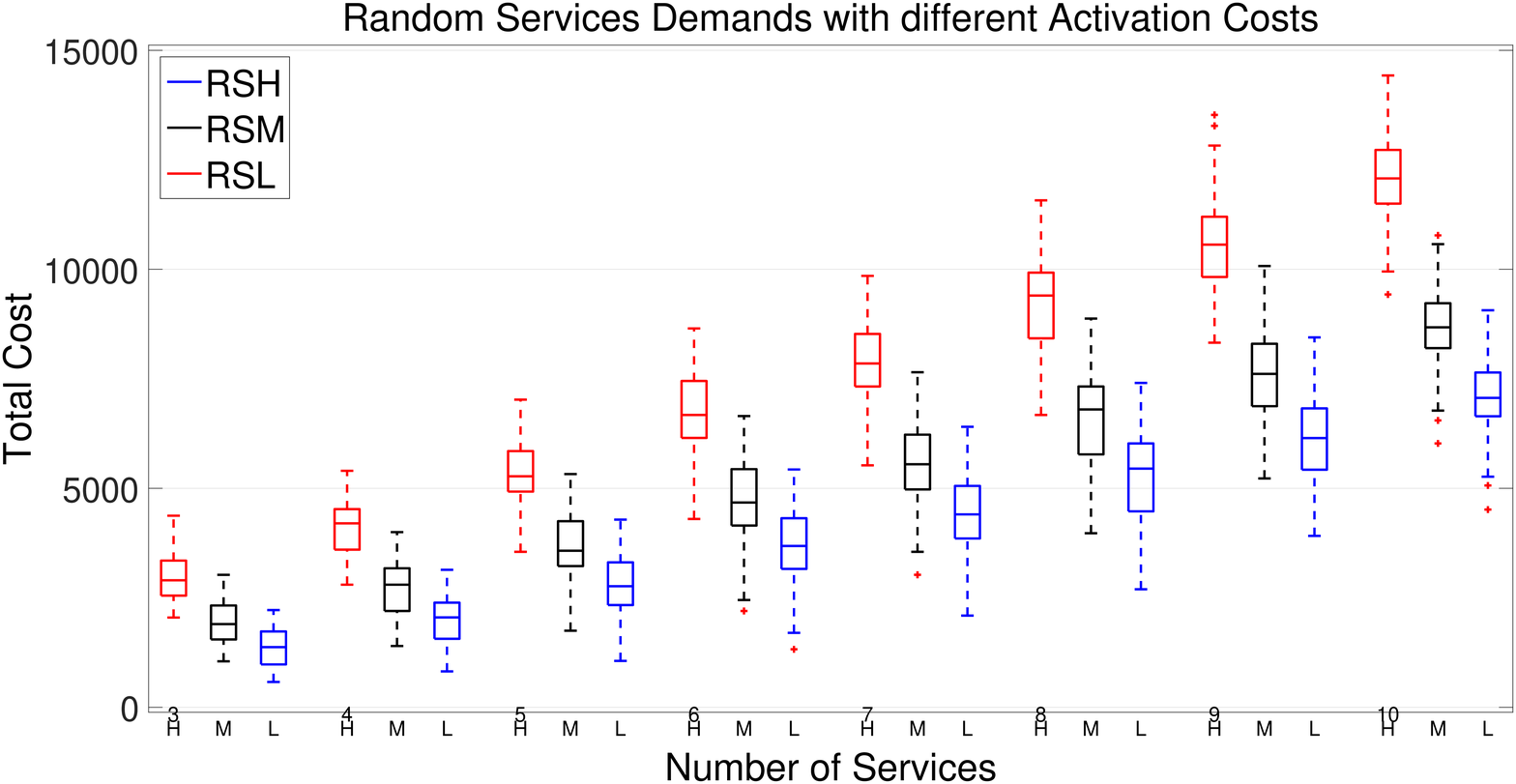}
	\caption{Total Cost vs Number of \textit{Random} Services. The figure shows the box-plots of the total allocation cost for a set of three to ten \textit{Random Services}. The reflection on the cost for three different values of activation cost is depicted. \textit{High }(red), \textit{Mixed} (black), and \textit{Low} (blue) activation cost is denoted with H, M, and L respectively.}
	\label{fig:Optimal_Cost_Random}
\end{figure}

\begin{figure}
	\centering
	\includegraphics[width=1\columnwidth]{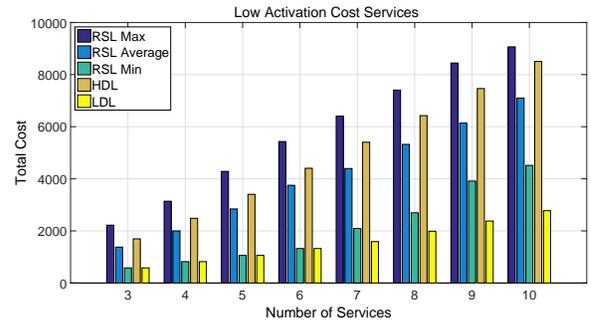}
	\caption{Total Cost vs Number of \textit{Random}, \textit{High}, and \textit{Low} Services. The figure shows the total allocation cost of three to ten services with interfaces having \textit{Low} activation cost. Maximum, Average, and Minimum total costs for \textit{Random Services} are depicted. }
	\label{fig:Optimal_Cost_HighAndLow}
\end{figure}

In \figurename{~\ref{fig:Splits_Optimal}} splits per service are presented for the same sets of services and activation costs we used previously (\tablename{ \ref{table:Simulation_Services}}). When no services are split, the number of splits per service is one.
Overall, when interfaces have a \textit{High} activation (in comparison to the utilization or lower activation) cost, splitting is not advantageous, since engaging additional interfaces results in paying a much higher total cost. An example can be seen in \figurename{~\ref{fig:Splits_Optimal}} where \textit{Random Services} with \textit{High}  or \textit{Mixed} activation cost, which is one order of magnitude higher than \textit{Low} activation cost, do not split on average. 

Conversely, keeping the activation cost \textit{Low} (compared to the utilization one) yields more splits per service --- especially when the utilization cost per unit is low. For instance, consider \textit{Random Services} with \textit{Low} activation cost. Such services split more to derive a benefit from interfaces with lower utilization cost per unit. Obviously, these interfaces' exploitation would not be beneficial to the total cost, if a higher activation cost was charged as well.

Furthermore, splits per service do not demonstrate a monotonic behavior in the number of services, assuming all used parameters are fixed except for the number of services. Consider, for instance, \textit{High} or \textit{Low} demand services in the same figure. 
When a relatively low number of \textit{High} demand services are used, more splits happen to exploit low utilization cost interfaces along with \textit{Low} activation cost. However, when more \textit{High} demand services are added and low utilization cost interfaces are depleted, then fewer splits per service occur. This behavior is attributed to the fact that low cost capacities are no more available and hence it is disadvantageous to split resources to interfaces with higher utilization cost (and be charged the corresponding activation cost as well). 

\begin{figure}
	\centering
	\includegraphics[width=1\columnwidth]{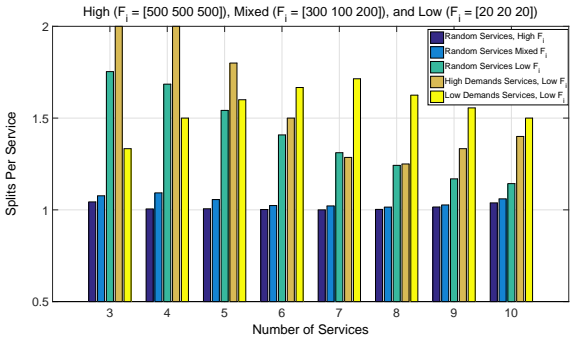}
	\caption{Splits per Service vs Number of Services for sets of services and activation costs of \tablename{~\ref{table:Simulation_Services}}.}
	\label{fig:Splits_Optimal}
\end{figure}

\subsection{Algorithms' Performance}\label{subSec:Algo_Solutions}
The results of \figurename{~\ref{fig:Costs_Algos} reveal that when the activation cost is low ($F_{i} = [20, 20, 20]$) compared with the utilization cost (top plots), the proposed algorithms closely approximate the optimal solution.
On the other hand, when the activation cost becomes one order of magnitude higher, as in the bottom plots of the figure, the algorithms' approximation is not as good.

This behavior can be attributed to the algorithms' allocation policy. Both algorithms initially attempt to fit the chosen demand to the two interfaces that are the least expensive ones, in terms of the utilization cost (see Lines 16-19 of Algorithm \ref{alg:avgCost}). Namely, if there are enough capacities on one of them, the algorithms will not take into consideration the corresponding activation cost. Hence, in case the chosen interface is not yet activated, a substantial extra cost may be charged. 

Note that in this case the optimal allocation may be another one: a service split may have proven more beneficial though the algorithms do not consider it. There is a significant computational advantage of this allocation policy, however. 
The algorithms do not search exhaustively for the best (i.e., lowest-cost) combination of allocations; two checks are enough on average. 

Plots in \figurename{~\ref{fig:Costs_Algos} also present the fact that neither of the two algorithms clearly outperforms the other. The \algoRandInit algorithm approaches the optimal cost better when relatively high activation cost interfaces are operating (bottom figures). Therefore, from an implementation perspective, the algorithm with the faster initialization should be preferred (recall that the two methods differ only in that step --- not the allocation policy they use).
		
\begin{figure*}[!htbp]
	\centering
	\subfloat{\includegraphics[width=.5\textwidth]{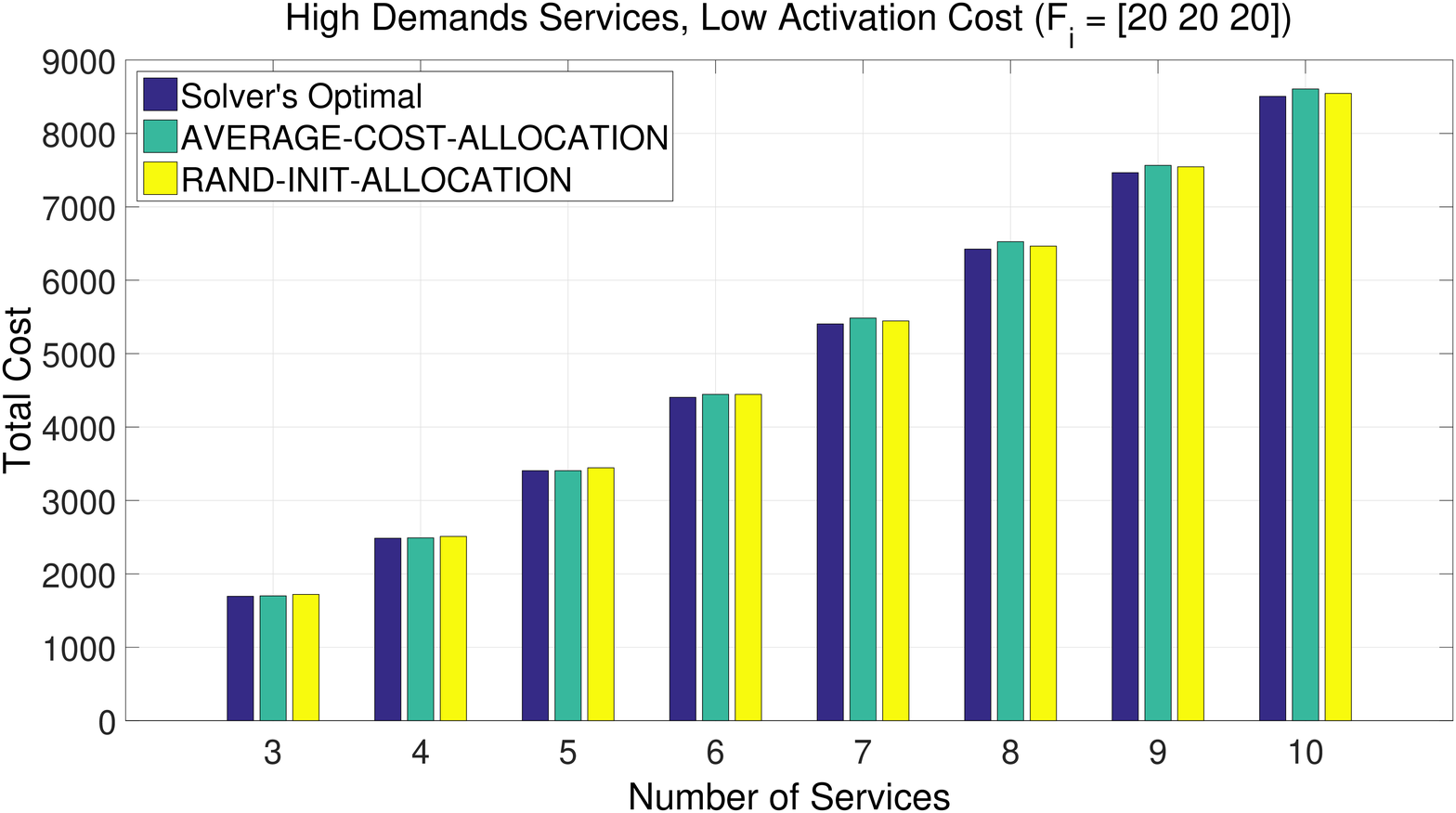}} 	\subfloat{\includegraphics[width=.5\textwidth]{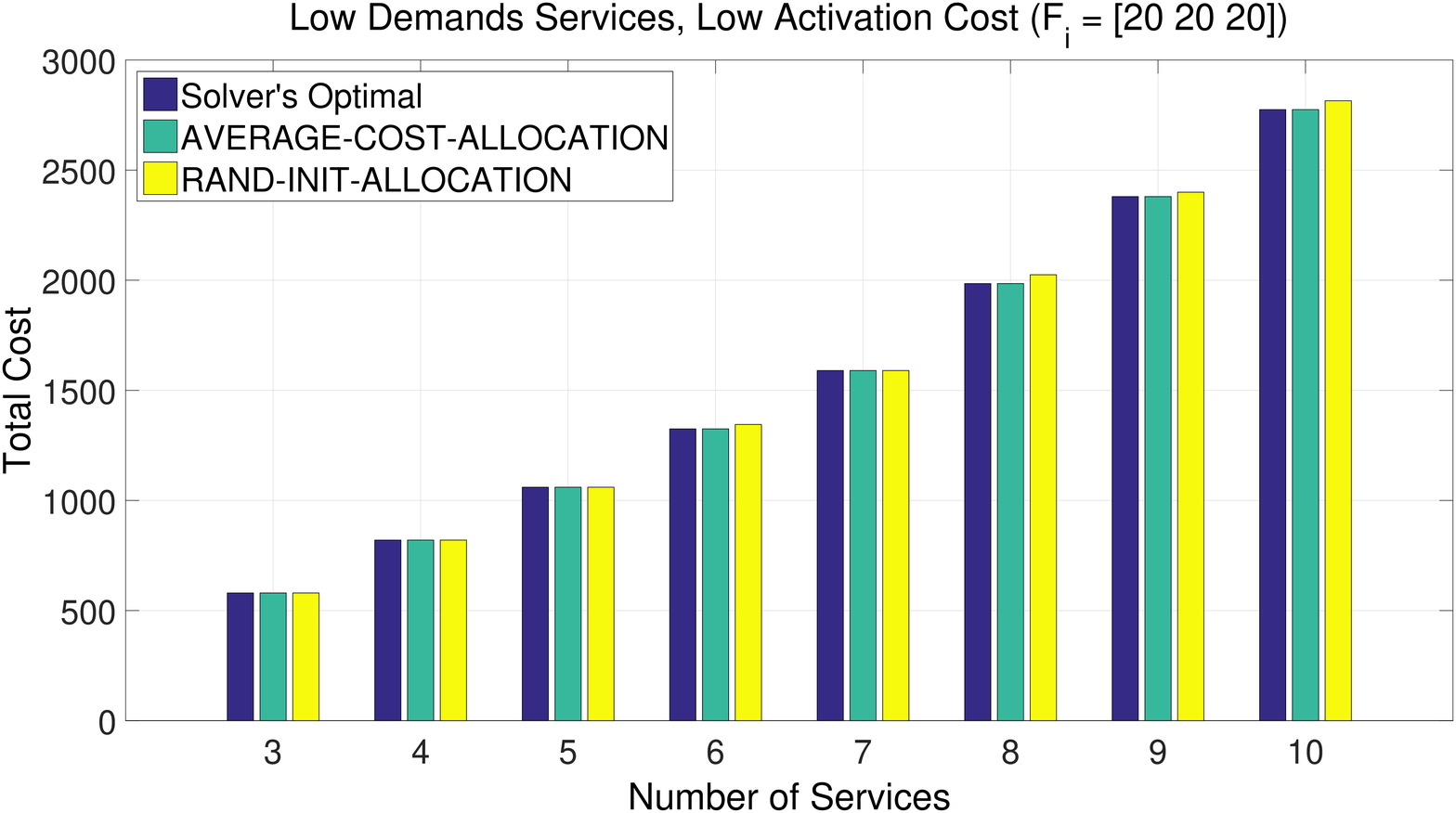}}\\	
	\subfloat{\includegraphics[width=.5\textwidth]{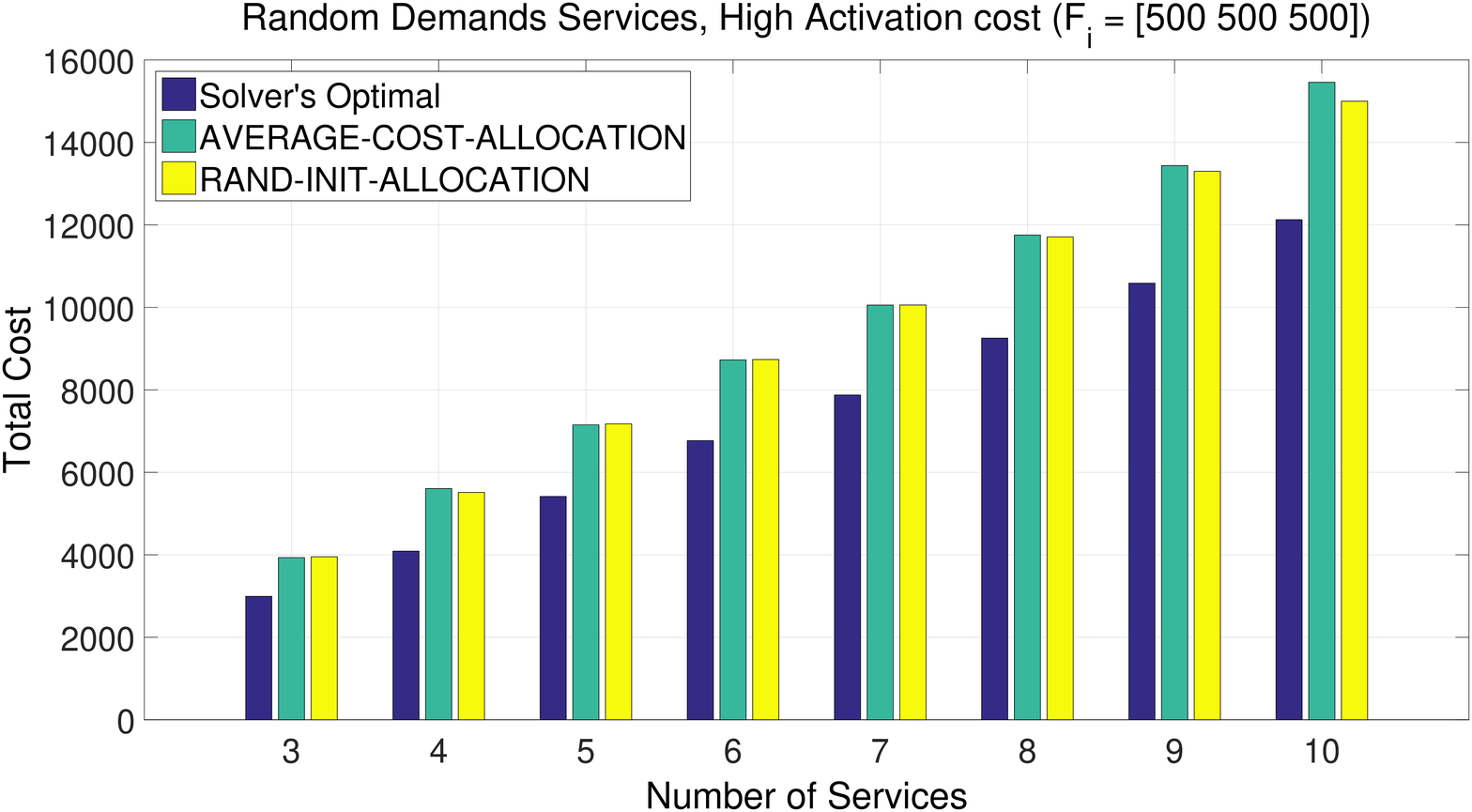}} 	
	\subfloat{\includegraphics[width=.5\textwidth]{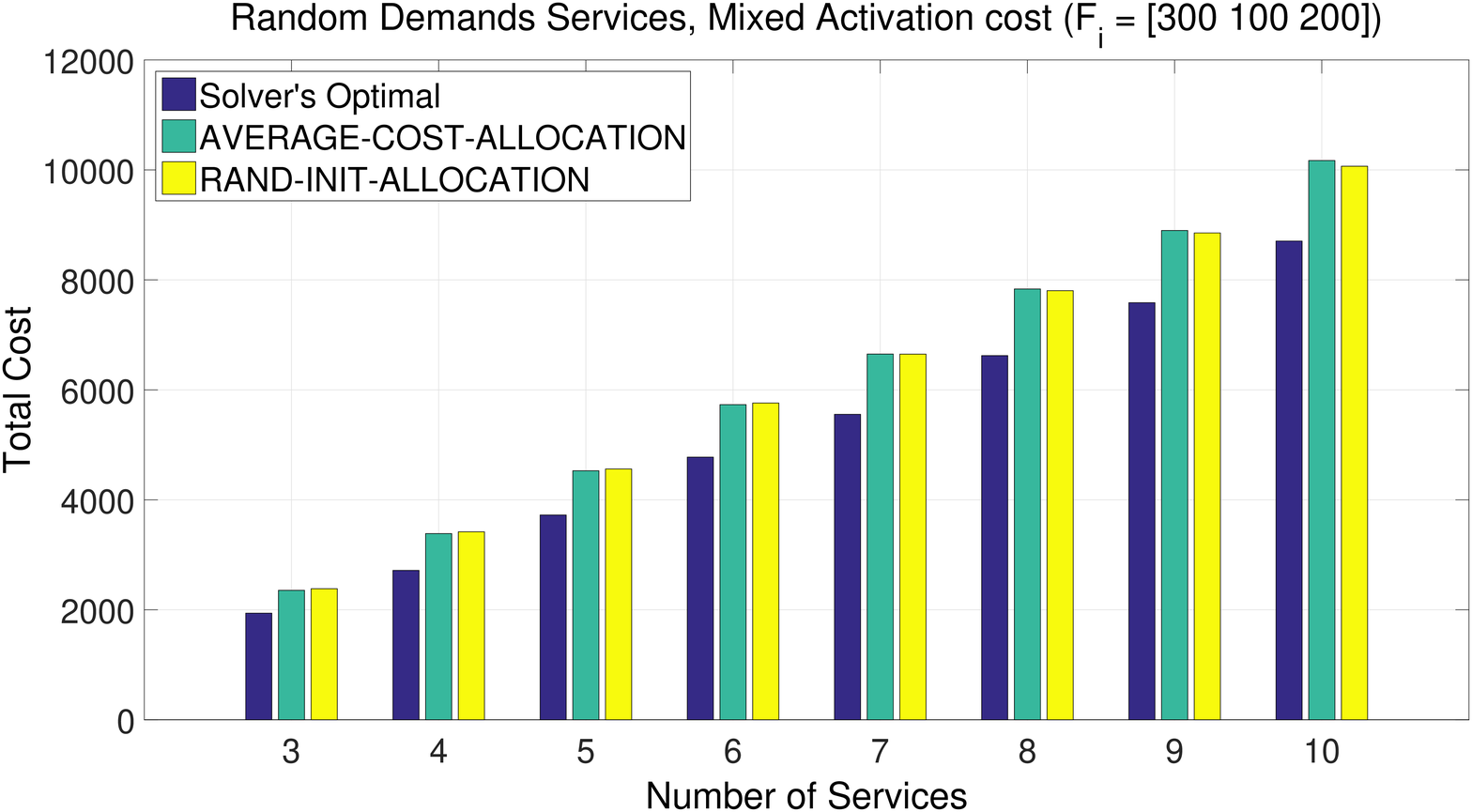}} 

	\caption{ Total Cost vs Number of Services for the algorithms we developed to approximate the optimal total cost: (i) \algoAvgCost begins allocating resources according to their average cost first, and (ii) \algoRandInit which begins allocating resources choosing randomly which service demands to allocate first.}
	\label{fig:Costs_Algos}
\end{figure*}

\subsection{Allocation over Multiple Rounds}\label{sec:SchedulingResults}

We performed several sets of simulations in Matlab to gain insight into the multi-round solution we discussed through the \textit{multi-round SIA} problem formulation. We performed the following setup:
we considered an IoT device of two interfaces with ten units of \res{1} and eight units of \res{2} each. The per unit utilization cost for the first interface was $c_{1} = (22,20)$ for  (\res{1}, \res{2}) respectively. The corresponding per unit cost for the second interface was set at $c_{2} = (20,8)$. Additionally, the activation costs of the two interfaces were $F_{1} = 100$ and $F_{2} = 110$.

We configured our simulator to allocate three, six, and nine services requiring different resource demands of great heterogeneity in a random manner. For example, one service's dominant share may be \res{1}, another's dominant share may be \res{2}, whilst a third one may demand the same share of both resources. 

We varied the number of rounds $R$ to gain insight into the cost sensitivity. The results can be found in \figurename~\ref{fig:Scheduling_costs}.
The necessity for more than one round is evident. For instance, the demands of nine services need at least five rounds to be fully served (hence in this case $R_{min} = 5$).

As expected, as we increase the number of rounds over which SIA takes place, the total cost is decreased. The maximum total cost for a specific configuration is incurred when $R_{min}$ is used. As we explained, this is anticipated since by increasing $R$, lower cost resources become available and hence used. As a result, the total cost is decreased as we add more rounds.

Furthermore, using more than $R_{max}$ rounds is not beneficial, since we have already exploited the lowest-cost interface for each resource in each round. For example, it is clear in \figurename~\ref{fig:Scheduling_costs} that it makes no sense using four or more rounds to serve three services; no decrease in cost will be observed. 

On the contrary, if we aim to reduce the total leftover capacities to exploit the interfaces' available resources and serve the demands in fewer rounds, we would rather use the smallest possible number of rounds ($R_{min}$). In this case we increase resource utilization in each round and the total cost becomes the highest possible at the same time.

\begin{figure}[!t]
\centering
\includegraphics[width=1\columnwidth]{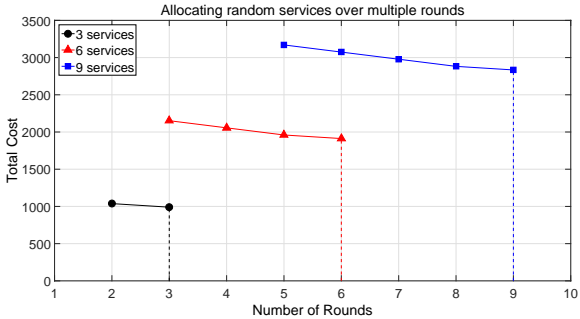}
\caption{Total Cost versus Number of Rounds. An example of three, six and nine services being allocated over different number of rounds. For three services (black line) the minimum and the maximum number of required rounds are $R_{min} = 2$ and $R_{max} = 3$ respectively. For six services (red line) $R_{min} = 3$ and $R_{max} = 6$, and for nine services (blue line) $R_{min} = 5$ and $R_{max} = 9$. Note the monotonicity of the total cost; more rounds yield lower total allocation cost.}
\label{fig:Scheduling_costs}
\end{figure}

\section{Conclusion}\label{sec:Conclusion}

We have introduced a solution to the problem of assigning services with heterogeneous and non-interchangeable resource demands to multiple network interfaces of an Internet of Things (IoT) device, while simultaneously minimizing the cost of using the interfaces. The total cost consists of the utilization cost that is charged for each served resource unit as well as the activation cost of each interface that is the cost of engaging an interface to serve a resource demand. We call this the \textit{Service-to-Interface Assignment (SIA)} problem. 

The solution we suggest is a precise mathematical formulation, which we have proved is NP-Complete. 
We have devised two \textit{SIA} versions. In the first one, the interfaces' available resources can serve the whole set of demands in one round.
We find the solver's optimal solution to the proposed formulation, which acts as a benchmark for the two algorithms we developed and presented to approximate the optimal solution. We have evaluated the proposed algorithms and shown under which circumstances they can approximate the optimal solution well. 
In the second \textit{SIA} version, when the resource demands exceed the IoT device's available resources in one round, we suggest formulating \textit{SIA} in multiple rounds. We use an appropriate number of rounds to optimize the resource allocations, while obtaining the minimum total cost of using the IoT device's interfaces at the same time. We call this the \textit{multi-round SIA} problem.

The numerical results show the role of the activation cost in the services' splits and distribution among the interfaces. Therefore, they can act as a guide for the design and implementation of real IoT applications and parameters e.g., to simulate the power drain of a battery-operated IoT device or change an applications' scheduling policy on-the-fly. A further contribution can be found in the \textit{multi-round SIA}'s results which demonstrate the effect of the number of rounds on the total cost, depending on the used policy. The difference of the cost between the policies of the two bounds (minimum rounds vs minimum cost) is more prominent when the number of services is increased. 

\bibliographystyle{ieeetr}
\bibliography{IoT_Journal_camera_ready}

\end{document}